\documentclass[11pt]{article}

\usepackage[T1]{fontenc}
\usepackage[utf8]{inputenc}
\usepackage[margin=1in]{geometry}
\usepackage{amsmath,amssymb,amsthm,mathtools}
\usepackage{enumitem}
\usepackage[hidelinks]{hyperref}
\usepackage[nameinlink,capitalize]{cleveref}
\usepackage[textsize=footnotesize]{todonotes}
\usepackage{xcolor}

\newtheorem{theorem}{Theorem}[section]

\newtheorem{corollary}[theorem]{Corollary}
\newtheorem{proposition}[theorem]{Proposition}
\theoremstyle{definition}
\newtheorem{definition}[theorem]{Definition}
\theoremstyle{remark}
\newtheorem{remark}[theorem]{Remark}
\theoremstyle{example}
\newtheorem{example}[theorem]{Example}

\newcommand{\F}{\mathbb{F}}
\newcommand{\Fp}{\mathbb{F}_p}

\newcommand{\Tr}{\operatorname{Tr}}
\newcommand{\AC}{\operatorname{AC}}
\newcommand{\DLCT}{\operatorname{DLCT}}
\newcommand{\DDT}{\operatorname{DDT}}

\newcommand{\1}{\mathbf{1}}

\newcommand{\Kl}{\operatorname{Kl}}

\newcommand*\colvec[3][]{
\begin{pmatrix}
\ifx\relax#1\relax\else#1\\\fi#2\\#3
\end{pmatrix}
}
\crefname{lemma}{Lemma}{Lemmas}
\Crefname{lemma}{Lemma}{Lemmas}
\crefname{proposition}{Proposition}{Propositions}
\Crefname{proposition}{Proposition}{Propositions}
\crefname{corollary}{Corollary}{Corollaries}
\Crefname{corollary}{Corollary}{Corollaries}
\crefname{theorem}{Theorem}{Theorems}
\Crefname{theorem}{Theorem}{Theorems}
\crefname{definition}{Definition}{Definitions}
\Crefname{definition}{Definition}{Definitions}
\crefname{remark}{Remark}{Remarks}
\Crefname{remark}{Remark}{Remarks}
\crefname{example}{Example}{Examples}
\Crefname{example}{Example}{Examples}

\title{\bfseries Differential-linear profiles over finite fields of arbitrary characteristic}
\author{\Large\bfseries Kirpa Garg$^1$, Constanza Riera$^2$ and Pantelimon St\u{a}nic\u{a}$^3$\\[6pt]
$^1$University of Rouen Normandy -- LITIS UR 4108,
Avenue de l'Universit\'e, \\
76800 Saint-\'Etienne-du-Rouvray, France;
\texttt{kirpa.garg@univ-rouen.fr}\\[4pt]
$^2$Department of Computer Science, Electrical Engineering and Mathematical Sciences,\\
Western Norway University of Applied Sciences, 5020 Bergen, Norway;
\texttt{csr@hvl.no}\\[4pt]
$^3$Applied Mathematics Department, Naval Postgraduate School,\\
Monterey, CA 93943, USA;
\texttt{pstanica@nps.edu}}
\date{}

\begin{document}
\maketitle

\begin{abstract}
The (binary) differential-linear connectivity table (DLCT) measures the dependence between an input difference and a linear mask applied to the corresponding output difference.  For vectorial Boolean functions, each DLCT entry is one half of an additive autocorrelation value.  We extend this relation to functions over finite fields of arbitrary prime characteristic by introducing a level-resolved $p$-ary differential-linear profile.  Its entries are the centered numbers of inputs for which a derivative component has each prescribed trace value in $\Fp$.  The discrete Fourier transform of this profile is the family of additive autocorrelations obtained by multiplying the output mask by the nonzero elements of $\Fp$; when $p=2$, the usual binary identity is recovered.

For a fixed input difference, we prove that the profiles over all nonzero output masks determine the corresponding DDT row exactly, and we give an explicit inversion formula.  We also establish a second-moment identity: the total profile energy in one derivative direction is a constant multiple of the squared Euclidean distance between that DDT row and the balanced row.  Thus this energy is determined by the full row differential spectrum, not by differential uniformity alone.  It follows that all profiles in a direction vanish exactly when the derivative is balanced; for square maps in odd characteristic, this gives a characterization of planarity.  As concrete odd-characteristic examples, we determine the complete profile of the monomial $x^{p^k+1}$ and derive an exact Kloosterman-sum formula for the inverse monomial.  Finally, we determine the behavior of the profiles under equivalence.  EA-equivalence reindexes the input and output masks and translates the trace level, whereas a general CCZ equivalence may mix several derivative directions.  Nevertheless, for square maps the global nontrivial profile energy is CCZ-invariant.
\end{abstract}

\noindent
{\bf Keywords.} differential-linear connectivity; $p$-ary functions; difference distribution table; autocorrelation; finite fields; CCZ equivalence.

\noindent{\bf Mathematics Subject Classification 2020.} 
{Primary 94A60; Secondary 11T71, 94D10.}

\section{Introduction}

Differential and linear cryptanalysis are two fundamental techniques for evaluating the security of block ciphers.  Differential-linear cryptanalysis combines them by following a differential through one part of a cipher and a linear approximation through another.  The interaction between these two parts need not be independent.  To measure this dependence, Bar-On, Dunkelman, Keller and Weizman introduced the differential-linear connectivity table (DLCT) at EUROCRYPT 2019~\cite{BDKW19}.

For a vectorial Boolean function $F:\F_2^n\to\F_2^m$, the DLCT entry associated with an input difference $u\in\F_2^n$ and a nonzero output mask $v\in\F_2^m$ measures the bias of the Boolean derivative component $x\longmapsto v\cdot\bigl(F(x+u)+F(x)\bigr)$. Subsequent work developed the binary theory substantially.  Canteaut, K\"olsch, Li, Li, Li, Qu and Wiemer~\cite{CanteautEtAl19} related the DLCT to additive autocorrelation, the Walsh transform and the differential distribution table (DDT), studied its behavior under affine, EA and CCZ equivalence, and treated APN, plateaued and almost bent functions together with several polynomial families.  Related binary structural results were obtained independently in~\cite{LiLiLiQu19}.  In particular, the binary DLCT is one half of the additive autocorrelation.  More recent work has produced additional binary families with low differential-linear uniformity~\cite{XieEtAl26}.

Differential-linear cryptanalysis has also been considered outside the binary setting.  Xu, Chen, Wang and Wei formalized differential-linear cryptanalysis over $\Fp$, used a nonbinary DLCT, and related it to the DDT in their analysis of an MPC-friendly cipher~\cite{XuEtAl23}.  More recently, Niu, Sun, Yan and Wang formulated differential-linear approximations over arbitrary finite abelian groups in terms of group characters~\cite{NiuEtAl26}.  These works provide the cryptanalytic and character-theoretic setting for nonbinary differential-linear analysis.  The present paper addresses a different structural question: how to retain, in a real integer-valued table, the separate multiplicities of all $p$ trace values of a derivative component.

Several related connectivity and higher-order differential notions have likewise been developed over finite fields of arbitrary characteristic.  Garg, Hasan, Riera and St\u{a}nic\u{a} studied second-order zero differential spectra for APN and other low-differential-uniformity functions in odd characteristic~\cite{GargEtAl25SOZD,GargEtAl25JAA}.  They showed, in particular, that for an odd or even function in odd characteristic, second-order zero differential uniformity one implies APN-ness, while the converse need not hold~\cite{GargEtAl25JAA}.  In separate work, the same authors related revised boomerang connectivity tables to the DDT and used this relation to obtain structural results and explicit values for several low-differential-uniformity families~\cite{GargEtAl25BCT}.  These results motivate developing the profile directly from the DDT and testing it on functions whose derivative structure is already understood.

For a vectorial $p$-ary function $F:\F_{p^n}\longrightarrow\F_{p^m}$, where $p$ is an arbitrary prime, one direct analogue of a binary DLCT entry is obtained by counting the inputs for which the trace of a derivative component is zero and subtracting the uniform value $p^{n-1}$.  For $p>2$, however, the trace can take $p$ values rather than only two.  The zero-trace count therefore retains only part of the level-distribution information that is contained in the binary entry.

We retain all trace levels.  For $j\in\Fp$, we center by $p^{n-1}$ the number of $x\in\F_{p^n}$ satisfying $\Tr_1^m(v(F(x+u)-F(x)))=j$.  We call the resulting $p$-tuple the $p$-ary differential-linear profile at $(u,v)$.  Its discrete Fourier transform is the family of additive autocorrelations obtained from the masks $tv$, $t\in\Fp^*$, where $\zeta_p=e^{2\pi i/p}$ is used for the canonical additive character.  Consequently, the familiar binary relation between DLCT and autocorrelation is not lost in odd characteristic; rather, it becomes a finite Fourier relation involving all nonzero scalar multiples of the mask.

The main results are as follows.  First, character orthogonality gives exact Fourier formulas between the profile and the autocorrelations along the scalar orbit $tv$, $t\in\Fp^*$.  Second, every profile entry is an affine trace-hyperplane sum of a DDT row, and the profiles over all nonzero output masks recover that row by an explicit inversion formula.  Third, the total squared profile values in a fixed derivative direction equal a constant multiple of the squared Euclidean distance between the corresponding DDT row and the balanced row.  This identity characterizes balanced derivatives by vanishing profile and, for square maps in odd characteristic, characterizes planar functions by the vanishing of all nontrivial profiles.  We then work out two odd-characteristic power-function examples: the monomial $x^{p^k+1}$ has either vanishing profiles or a single explicit extreme pattern, according to the parity of $n/\gcd(k,n)$, while the inverse monomial has an exact profile formula in terms of classical Kloosterman sums.  Finally, EA-equivalence reindexes the masks and translates the trace level, so it preserves the complete profile spectrum and $\Gamma_F^{(p)}$.  For square maps, an exact CCZ transport law explains how a general graph equivalence can mix derivative directions; although individual profiles need not be preserved, the global nontrivial profile energy is CCZ-invariant.

\medskip
\noindent\textbf{Literature note.}
The standard DLCT and its structural theory are binary~\cite{BDKW19,CanteautEtAl19,LiLiLiQu19}.  Xu et al.~\cite{XuEtAl23} extended differential-linear analysis and the DLCT to $\Fp$ and related the resulting table to the DDT, while Niu et al.~\cite{NiuEtAl26} treated differential-linear approximations over finite abelian groups through group characters.  A $c$-DLCT for vectorial Boolean functions has also been introduced and studied~\cite{EddahmaniMesnager24}, and we do not pursue that extension here.  Our contribution is the level-resolved, integer-valued profile and the consequences of retaining all trace levels: DDT-row reconstruction, exact moment identities, and the equivalence results proved below.  The basic Fourier step itself is a direct application of character orthogonality.

The paper is organized as follows.  Section~\ref{S2} recalls the binary DLCT and additive autocorrelation.  Section~\ref{S3} defines the $p$-ary profile and establishes its Fourier representation.  Section~\ref{sec:ddt-structure} develops the DDT reconstruction and moment identities and then treats the Gold-type and inverse monomials.  Section~\ref{sec:equivalence} studies EA and CCZ equivalence, and Section~\ref{S6} summarizes the results and open directions.

\section{Background and the binary setting}\label{S2}

Throughout the paper, $p$ is a prime and $n,m$ are positive integers.  We write $\F_{p^r}$ for the finite field of $p^r$ elements, $\F_{p^r}^*=\F_{p^r}\setminus\{0\}$, and $\Tr_1^r(z)=z+z^p+\cdots+z^{p^{r-1}}$ for the absolute trace from $\F_{p^r}$ to $\Fp$.  We set $\zeta_p=e^{2\pi i/p}$. When there is no ambiguity, we write simply $\Tr$ for the relevant absolute trace.  For a condition $P$, $\1_{\{P\}}$ denotes its indicator, equal to $1$ when $P$ holds and to $0$ otherwise; equivalently, for a set $S$, $\1_S$ denotes its characteristic function.

Let $F:\F_{p^n}\longrightarrow\F_{p^m}$ be a vectorial $p$-ary function.  For $u\in\F_{p^n}$, its additive derivative in direction $u$ is
\begin{equation*}
 D_uF(x)=F(x+u)-F(x).
\end{equation*}
For $p=2$, subtraction and addition coincide, and $D_a$ is the usual derivative of a vectorial Boolean function.

The differential distribution table of $F$ is defined by
\begin{equation*}
 \DDT_F(u,b)=\#\{x\in\F_{p^n}:D_uF(x)=b\},
 \qquad u\in\F_{p^n},\ b\in\F_{p^m}.
\end{equation*}
The differential uniformity is
\[
 \delta_F=\max_{u\in\F_{p^n}^*,\,b\in\F_{p^m}}\DDT_F(u,b).
\]
When $m=n$, we call $F$ almost perfect nonlinear (APN) if $\delta_F=2$.  When $p$ is odd and $m=n$, we call $F$ planar (or perfect nonlinear) if $D_uF$ is a permutation of $\F_{p^n}$ for every $u\in\F_{p^n}^*$.

For $v\in\F_{p^m}$, the component of the derivative selected by $v$ is the $p$-ary function
\begin{equation}\label{eq:component-derivative}
 x\longmapsto \Tr_1^m\bigl(vD_uF(x)\bigr).
\end{equation}
Its additive autocorrelation is naturally expressed with additive characters.

\begin{definition}\label{def:ac}
For $u\in\F_{p^n}$ and $v\in\F_{p^m}$, define
\begin{equation}\label{eq:ac}
 \AC_F(u,v)
 =\sum_{x\in\F_{p^n}}
 \zeta_p^{\Tr_1^m(vD_uF(x))}.
\end{equation}
\end{definition}

In general $\AC_F(u,v)$ is a cyclotomic integer and need not be real when $p>2$.  We have $\AC_F(0,v)=p^n\ \forall v\in\F_{p^m},\qquad \AC_F(u,0)=p^n\ \forall u\in\F_{p^n}$, and
\begin{equation}\label{eq:conjugate}
 \AC_F(u,-v)=\overline{\AC_F(u,v)}.
\end{equation}

For later use we recall the standard orthogonality relation
\begin{equation}\label{eq:orthogonality}
 \frac1p\sum_{t\in\Fp}\zeta_p^{t(a-b)}
 =
 \begin{cases}
 1,&a=b,\\
 0,&a\ne b,
 \end{cases}
 \qquad a,b\in\Fp.
\end{equation}

The preceding notation applies in all prime characteristics.  We now recall the usual binary DLCT in the same trace notation, both to fix the normalization and to make clear what changes when $p>2$.  Let $F:\F_{2^n}\to\F_{2^m}$.  For $u\in\F_{2^n}$ and $v\in\F_{2^m}$, the DLCT entry is
\begin{equation}\label{eq:binary-dlct}
 \DLCT_F(u,v)
 =\#\left\{x\in\F_{2^n}:
 \Tr_1^m\bigl(vD_uF(x)\bigr)=0\right\}-2^{n-1}.
\end{equation}
Equivalently, in vector-space notation the trace component may be replaced by the usual inner product on $\F_2^m$.

Since a Boolean function takes only the values $0$ and $1$, the two level counts sum to $2^n$.  Therefore the zero-level bias determines the full distribution of the derivative component.  From~\eqref{eq:ac},
\begin{align*}
 \AC_F(u,v)
 &=\#\{x:\Tr(vD_uF(x))=0\}
   -\#\{x:\Tr(vD_uF(x))=1\}\\
 &=2\DLCT_F(u,v).
\end{align*}
Thus
\begin{equation}\label{eq:binary-ac-dlct}
 \DLCT_F(u,v)=\frac12\AC_F(u,v).
\end{equation}
This identity is the starting point of the binary spectral theory.  In particular, Canteaut et al.~\cite{CanteautEtAl19} use the autocorrelation formulation to derive Walsh-transform and DDT characterizations, moment identities and bounds, and to analyze equivalence properties.  These results provide the principal binary benchmark for the $p$-ary constructions below.

The binary differential-linear uniformity is commonly defined by
\begin{equation*}
 \gamma_F=
 \max_{u\in\F_{2^n}^*,\,v\in\F_{2^m}^*}
 \left|\DLCT_F(u,v)\right|.
\end{equation*}

\section{The \texorpdfstring{$p$}{p}-ary differential-linear profile and its Fourier representation}\label{S3}

For $p=2$, the zero-level count determines the entire distribution because the complementary level is forced.  This ceases to be true when $p>2$: the trace component~\eqref{eq:component-derivative} has $p$ possible values, and a single level count does not determine the others.  This observation suggests retaining the full distribution of the trace values rather than privileging the zero level from the outset.

\begin{definition}\label{def:level-count}
For $u\in\F_{p^n}$, $v\in\F_{p^m}$, and $j\in\Fp$, let
\begin{equation*}
 N_F(u,v;j)
 =\#\left\{x\in\F_{p^n}:
 \Tr_1^m\bigl(vD_uF(x)\bigr)=j\right\}.
\end{equation*}
The centered level count is the $p$-ary differential-linear table:
\begin{equation*}
 L_F(u,v;j)=N_F(u,v;j)-p^{n-1}.
\end{equation*}
We call $\mathcal L_F(u,v)=\bigl(L_F(u,v;j)\bigr)_{j\in\Fp}$ the \emph{$p$-ary differential-linear profile} of $F$ at $(u,v)$.
\end{definition}

Since the $p$ level sets in Definition~\ref{def:level-count} partition $\F_{p^n}$,
$\sum_{j\in\Fp}N_F(u,v;j)=p^n$.  Therefore
\begin{equation}\label{eq:profile-sum-zero}
 \sum_{j\in\Fp}L_F(u,v;j)=0.
\end{equation}
Thus the profile has at most $p-1$ independent entries.

The direct analogue of the usual binary DLCT is the zero-level entry.

\begin{definition}
\label{def:p-dlct-zero}
Define the zero-level $p$-ary DLCT entry by
\begin{equation}\label{eq:p-dlct-zero}
 \DLCT_F^{(p)}(u,v)
 :=L_F(u,v;0)
 =\#\left\{x\in\F_{p^n}:
 \Tr_1^m\bigl(vD_uF(x)\bigr)=0\right\}-p^{n-1}.
\end{equation}
\end{definition}

For $p=2$, \Cref{def:p-dlct-zero} is exactly~\eqref{eq:binary-dlct}.  For $p>2$, however, the zero-level entry alone does not determine the other $p-1$ level counts, so the binary identity with a single autocorrelation no longer holds.  The next proposition gives the precise Fourier relation between the full $p$-ary profile and the corresponding additive autocorrelations.

\begin{proposition}\label{prop:fourier-relation}
For every $u\in\F_{p^n}$, $v\in\F_{p^m}$ and $j\in\Fp$,
\begin{equation}
\label{eq:profile-from-ac}
 L_F(u,v;j)
 =\frac1p\sum_{t\in\Fp^*}
 \zeta_p^{-tj}\AC_F(u,tv).
\end{equation}
Conversely, for every $t\in\Fp^*$,
\begin{equation}
\label{eq:ac-from-profile}
 \AC_F(u,tv)
 =\sum_{j\in\Fp}L_F(u,v;j)\zeta_p^{tj}.
\end{equation}
\end{proposition}

\begin{proof}
Fix $u$, $v$, and $j$.  Applying the orthogonality relation~\eqref{eq:orthogonality} with $a=\Tr_1^m\bigl(vD_uF(x)\bigr), b=j$, we obtain, for every $x\in\F_{p^n}$,
\[
 \1_{\{\Tr(vD_uF(x))=j\}}
 =\frac1p\sum_{t\in\Fp}
 \zeta_p^{t(\Tr(vD_uF(x))-j)}.
\]
Summing this identity over $x$ and interchanging the two finite sums yields
\begin{align*}
 N_F(u,v;j)
 &=\frac1p\sum_{t\in\Fp}\zeta_p^{-tj}
   \sum_{x\in\F_{p^n}}
   \zeta_p^{\Tr(tvD_uF(x))}.
\end{align*}
The contribution of $t=0$ is $\frac1p\sum_{x\in\F_{p^n}}1 =\frac{p^n}{p}=p^{n-1}$. For $t\ne0$, the inner sum is, by Definition~\ref{def:ac}, exactly $\AC_F(u,tv)$.  Hence
\[
 N_F(u,v;j)
 =p^{n-1}+\frac1p\sum_{t\in\Fp^*}
 \zeta_p^{-tj}\AC_F(u,tv).
\]
Using $L_F(u,v;j)=N_F(u,v;j)-p^{n-1}$, we get
\[
 L_F(u,v;j)
 =\frac1p\sum_{t\in\Fp^*}
 \zeta_p^{-tj}\AC_F(u,tv),
\]
which is~\eqref{eq:profile-from-ac}.

For the converse relation, group the terms in the autocorrelation sum according to the value of the trace component.  For $t\in\Fp^*$,
\begin{align*}
 \AC_F(u,tv)
 &=\sum_{j\in\Fp}N_F(u,v;j)\zeta_p^{tj}\\
 &=\sum_{j\in\Fp}\bigl(L_F(u,v;j)+p^{n-1}\bigr)\zeta_p^{tj}.
\end{align*}
Since $t\ne0$, the additive character $j\mapsto\zeta_p^{tj}$ is nontrivial and hence $\sum_{j\in\Fp}\zeta_p^{tj}=0$. The constant term therefore vanishes, leaving $\AC_F(u,tv)=\sum_{j\in\Fp}L_F(u,v;j)\zeta_p^{tj}$, which is~\eqref{eq:ac-from-profile}.
\end{proof}

Taking $j=0$ in (\ref{eq:profile-from-ac}) yields the promised replacement for the binary identity~\eqref{eq:binary-ac-dlct}.

\begin{corollary}\label{cor:zero-level-ac}
For all $u\in\F_{p^n}$ and $v\in\F_{p^m}$,
\begin{equation}
\label{eq:zero-level-ac}
 \DLCT_F^{(p)}(u,v)
 =\frac1p\sum_{t\in\Fp^*}\AC_F(u,tv).
\end{equation}
In particular, for $p=2$ this reduces to $\DLCT_F(u,v)=\frac12\AC_F(u,v)$. \end{corollary}
\begin{proof}
By Definition~\ref{def:p-dlct-zero}, $\DLCT_F^{(p)}(u,v)=L_F(u,v;0)$.  Setting $j=0$ in~\eqref{eq:profile-from-ac} gives $L_F(u,v;0) =\frac1p\sum_{t\in\Fp^*}\AC_F(u,tv)$, which is~\eqref{eq:zero-level-ac}.  If $p=2$, then $\Fp^*=\{1\}$, and the sum contains the single term $\AC_F(u,v)$.
\end{proof}

\begin{remark}
\label{rem:orbit}
For $p>2$, the natural autocorrelation object associated with the one-dimensional $\Fp$-subspace $\Fp v\subseteq\F_{p^m}$ is the collection $\bigl(\AC_F(u,tv)\bigr)_{t\in\Fp^*}$, not a single coefficient.  \Cref{prop:fourier-relation} shows that this collection and the differential-linear profile $\mathcal L_F(u,v)$ contain exactly the same information.
\end{remark}

\begin{remark}
\label{rem:real}
Although the individual values $\AC_F(u,tv)$ may be nonreal, the quantities $L_F(u,v;j)$ are integers.  In~\eqref{eq:profile-from-ac}, the terms indexed by $t$ and $-t$ are complex conjugates by~\eqref{eq:conjugate}, so the right-hand side is real, as required.
\end{remark}

The same Fourier description also yields an exact Parseval identity.

\begin{proposition}
\label{prop:parseval-profile}
For every $u\in\F_{p^n}$ and $v\in\F_{p^m}$,
\begin{equation}\label{eq:parseval-profile}
 \sum_{j\in\Fp}|L_F(u,v;j)|^2
 =\frac1p\sum_{t\in\Fp^*}|\AC_F(u,tv)|^2.
\end{equation}
\end{proposition}

\begin{proof}
For fixed $u$ and $v$, write $L_j=L_F(u,v;j)$ and define the Fourier coefficients of the centered profile by $B_t=\sum_{j\in\Fp}L_j\zeta_p^{tj},\qquad t\in\Fp$. By~\eqref{eq:profile-sum-zero}, the coefficient at $t=0$ is $B_0=\sum_{j\in\Fp}L_j=0$.

For $t\in\Fp^*$, Proposition~\ref{prop:fourier-relation} gives $B_t=\AC_F(u,tv)$. Applying Parseval directly to the finite Fourier transform of $(L_j)_{j\in\Fp}$ gives
\begin{align*}
 \sum_{t\in\Fp}|B_t|^2
 =\sum_{t\in\Fp}
   \left(\sum_{j\in\Fp}L_j\zeta_p^{tj}\right)
   \left(\sum_{k\in\Fp}L_k\zeta_p^{-tk}\right)
 =\sum_{j,k\in\Fp}L_jL_k
   \sum_{t\in\Fp}\zeta_p^{t(j-k)}
 =p\sum_{j\in\Fp}|L_j|^2,
\end{align*}
where the last equality follows from~\eqref{eq:orthogonality}.  Since $B_0=0$ and $B_t=\AC_F(u,tv)$ for every $t\in\Fp^*$, we obtain
\[
 \sum_{j\in\Fp}|L_F(u,v;j)|^2
 =\frac1p\sum_{t\in\Fp^*}|\AC_F(u,tv)|^2,
\]
as claimed.
\end{proof}

The Fourier description also makes clear that the relevant mask is naturally considered up to multiplication by a nonzero prime-field scalar.  At the level of the profile, such a multiplication simply relabels the trace values.

\begin{proposition}
\label{prop:scalar-profile}
Let $a\in\Fp^*$.  Then
\begin{equation}
\label{eq:scalar-profile}
 L_F(u,av;j)=L_F(u,v;a^{-1}j)
 \qquad (j\in\Fp).
\end{equation}
Consequently, the multiset $\{L_F(u,v;j):j\in\Fp\}$ depends only on the one-dimensional $\Fp$-subspace~$v\Fp$.
\end{proposition}

\begin{proof}
Since $a\in\Fp$ and the absolute trace is $\Fp$-linear,$ \Tr_1^m\bigl(avD_uF(x)\bigr)
 =a\Tr_1^m\bigl(vD_uF(x)\bigr)$.
Since $a\ne0$, multiplication by $a$ is a permutation of $\Fp$, and therefore
\[
 \Tr_1^m\bigl(avD_uF(x)\bigr)=j
 \quad\Longleftrightarrow\quad
 \Tr_1^m\bigl(vD_uF(x)\bigr)=a^{-1}j.
\]
The two conditions define the same subset of $\F_{p^n}$ after relabeling the trace value, and therefore $N_F(u,av;j)=N_F(u,v;a^{-1}j)$. Centering both sides by the common uniform value $p^{n-1}$ gives
\[
 \begin{aligned}
 L_F(u,av;j)
 &=N_F(u,av;j)-p^{n-1}\\
 &=N_F(u,v;a^{-1}j)-p^{n-1}\\
 &=L_F(u,v;a^{-1}j),
 \end{aligned}
\]
which is~\eqref{eq:scalar-profile}.  Finally, the map $j\mapsto a^{-1}j$ is a permutation of $\Fp$, so the two profile multisets coincide.
\end{proof}

The profile can be compressed into several natural differential-linear parameters.  For $p>2$, these quantities need not coincide, and the choice should ultimately reflect the intended cryptanalytic interpretation.  The quantity closest to the classical binary DLCT is obtained by retaining only the zero level,
\begin{equation}
\label{eq:zero-dlu}
 \gamma_F^{(p,0)}
 =\max_{u\in\F_{p^n}^*,\,v\in\F_{p^m}^*}
 \left|L_F(u,v;0)\right|.
\end{equation}
A definition intrinsic to the full profile is obtained by allowing all trace levels,
\begin{equation}
\label{eq:profile-dlu}
 \Gamma_F^{(p)}
 =\max_{u\in\F_{p^n}^*,\,v\in\F_{p^m}^*,\,j\in\Fp}
 \left|L_F(u,v;j)\right|.
\end{equation}
For $p=2$, the two definitions coincide because $L_F(u,v;1)=-L_F(u,v;0)$. For $p>2$, the two parameters need not coincide.  We retain both: $\gamma_F^{(p,0)}$ measures the distinguished zero level, whereas $\Gamma_F^{(p)}$ measures the full profile and, as shown in Section~\ref{sec:equivalence}, has the natural EA-invariance property.

Since the profile and the autocorrelation orbit are Fourier-equivalent, one may instead measure the latter directly, for example by
\begin{equation}
\label{eq:ac-orbit-parameter}
 A_F^{(p)}
 =\max_{u\in\F_{p^n}^*,\,v\in\F_{p^m}^*}
 \max_{t\in\Fp^*}|\AC_F(u,tv)|.
\end{equation}
The two Fourier formulas in \Cref{prop:fourier-relation} give explicit comparison bounds.  Indeed, from~\eqref{eq:profile-from-ac},
\[
 |L_F(u,v;j)|
 \le \frac1p\sum_{t\in\Fp^*}|\AC_F(u,tv)|
 \le \frac{p-1}{p}A_F^{(p)},
\]
whereas~\eqref{eq:ac-from-profile} gives
\[
 |\AC_F(u,tv)|
 \le \sum_{j\in\Fp}|L_F(u,v;j)|
 \le p\Gamma_F^{(p)}.
\]
Taking maxima in the two inequalities yields
\begin{equation*}
 \frac1p A_F^{(p)}
 \le \Gamma_F^{(p)}
 \le \frac{p-1}{p}A_F^{(p)}.
\end{equation*}
For $p=2$, both inequalities are equalities and recover the usual relation $\DLCT=\AC/2$.  For $p>2$, determining which normalization is most useful in a differential-linear attack requires a separate cryptanalytic analysis.

\section{DDT structure and moment identities}
\label{sec:ddt-structure}

The profile is closely tied to the differential distribution table.  Each level count is obtained by aggregating a DDT row over an affine trace hyperplane determined by the output mask, and the subsequent moment identity shows that this relation controls the total profile energy exactly.  The first identity makes the hyperplane-sum relation explicit.

\begin{proposition}
\label{prop:profile-ddt}
For every $u\in\F_{p^n}$, $v\in\F_{p^m}$, and $j\in\Fp$,
\begin{equation}
\label{eq:profile-ddt}
 N_F(u,v;j)
 =\sum_{\substack{b\in\F_{p^m}\\ \Tr_1^m(vb)=j}}
 \DDT_F(u,b).
\end{equation}
Consequently,
\begin{equation}
\label{eq:centered-profile-ddt}
 L_F(u,v;j)
 =\sum_{\substack{b\in\F_{p^m}\\ \Tr_1^m(vb)=j}}
 \DDT_F(u,b)-p^{n-1}.
\end{equation}
\end{proposition}

\begin{proof}
For fixed $u$, the sets $\{x\in\F_{p^n}:D_uF(x)=b\},\qquad b\in\F_{p^m}$, form a partition of $\F_{p^n}$.  The condition $\Tr_1^m\bigl(vD_uF(x)\bigr)=j$ holds precisely when the derivative value $b=D_uF(x)$ lies in the affine hyperplane $H_{v,j}=\{b\in\F_{p^m}:\Tr_1^m(vb)=j\}$. Therefore
\[
 \begin{aligned}
 N_F(u,v;j)
 &=\sum_{b\in H_{v,j}}
   \#\{x\in\F_{p^n}:D_uF(x)=b\}\\
 &=\sum_{\substack{b\in\F_{p^m}\\\Tr_1^m(vb)=j}}
   \DDT_F(u,b),
 \end{aligned}
\]
which is~\eqref{eq:profile-ddt}.  Since $L_F(u,v;j)=N_F(u,v;j)-p^{n-1}$, the same identity gives
\[
 L_F(u,v;j)
 =\sum_{\substack{b\in\F_{p^m}\\\Tr_1^m(vb)=j}}
   \DDT_F(u,b)-p^{n-1},
\]
namely~\eqref{eq:centered-profile-ddt}.
\end{proof}

Thus the $p$-ary differential-linear profile measures how far a DDT row is from being uniformly distributed across the parallel affine trace hyperplanes associated with $v$.  The relation is in fact lossless when all nonzero output masks are retained: the complete family of profiles for a fixed input difference determines the corresponding DDT row exactly.

\begin{theorem}
\label{thm:profile-ddt-inversion}
Fix $u\in\F_{p^n}$.  The complete collection $\{\mathcal L_F(u,v):v\in\F_{p^m}^*\}$ determines the $u$-th row of the differential distribution table.  More precisely, for every $b\in\F_{p^m}$,
\begin{equation}
\label{eq:profile-ddt-inversion}
 \DDT_F(u,b)
 =p^{n-m}+\frac1{p^m}
 \sum_{v\in\F_{p^m}^*}
 \zeta_p^{-\Tr_1^m(vb)}
 \sum_{j\in\Fp}L_F(u,v;j)\zeta_p^j.
\end{equation}
Conversely, the $u$-th DDT row determines every profile $\mathcal L_F(u,v)$ through~\eqref{eq:profile-ddt}.  Hence, for fixed $u$, the DDT row and the complete family of level-resolved profiles determine each other.
\end{theorem}

\begin{proof}
Fix $u\in\F_{p^n}$ and write $d_u(b)=\DDT_F(u,b),\qquad b\in\F_{p^m}$. By definition, $d_u(b)$ is the number of $x\in\F_{p^n}$ for which $D_uF(x)=b$.  Hence, for every $v\in\F_{p^m}$, we may partition the autocorrelation sum according to the value $b=D_uF(x)$:
\begin{align}
 \AC_F(u,v)
 =\sum_{x\in\F_{p^n}}
   \zeta_p^{\Tr_1^m(vD_uF(x))} 
 =\sum_{b\in\F_{p^m}}
   \sum_{\substack{x\in\F_{p^n}\\D_uF(x)=b}}
   \zeta_p^{\Tr_1^m(vb)} 
 =\sum_{b\in\F_{p^m}}
   d_u(b)\zeta_p^{\Tr_1^m(vb)}.
\label{eq:ac-ddt-fourier}
\end{align}

We now invert this finite Fourier transform directly.  For fixed $b_0\in\F_{p^m}$, multiply~\eqref{eq:ac-ddt-fourier} by
$\zeta_p^{-\Tr_1^m(vb_0)}$ and sum over $v\in\F_{p^m}$,
\begin{align*}
 \sum_{v\in\F_{p^m}}
 \AC_F(u,v)\zeta_p^{-\Tr_1^m(vb_0)}
 =\sum_{v\in\F_{p^m}}
   \sum_{b\in\F_{p^m}}
   d_u(b)\zeta_p^{\Tr_1^m(v(b-b_0))}
 =\sum_{b\in\F_{p^m}}d_u(b)
   \sum_{v\in\F_{p^m}}
   \zeta_p^{\Tr_1^m(v(b-b_0))}.
\end{align*}
The trace pairing $(v,c)\mapsto\Tr_1^m(vc)$ is nondegenerate.  Therefore the character
$v\mapsto\zeta_p^{\Tr_1^m(vc)}$ is trivial exactly when $c=0$.  By additive-character orthogonality,
\[
 \sum_{v\in\F_{p^m}}
 \zeta_p^{\Tr_1^m(vc)}
 =\begin{cases}
 p^m,&c=0,\\
 0,&c\ne0.
 \end{cases}
\]
Only the term $b=b_0$ survives, and consequently
\begin{equation}
\label{eq:ddt-fourier-inversion}
 d_u(b_0)
 =\frac1{p^m}\sum_{v\in\F_{p^m}}
 \AC_F(u,v)\zeta_p^{-\Tr_1^m(vb_0)}.
\end{equation}

We next separate the term $v=0$.  Since every summand in $\AC_F(u,0)$ is equal to $1$, $\AC_F(u,0)=p^n$, so its contribution to~\eqref{eq:ddt-fourier-inversion} is $p^{n-m}$.  For every $v\ne0$, Proposition~\ref{prop:fourier-relation}, with $t=1$, gives $\AC_F(u,v)=\sum_{j\in\Fp}L_F(u,v;j)\zeta_p^j$. We now split the sum in~\eqref{eq:ddt-fourier-inversion} into the term $v=0$ and the nonzero masks,
\[
 \begin{aligned}
 d_u(b_0)
 &=\frac1{p^m}\AC_F(u,0)
   +\frac1{p^m}\sum_{v\in\F_{p^m}^*}
    \AC_F(u,v)\zeta_p^{-\Tr_1^m(vb_0)}\\
 &=p^{n-m}
   +\frac1{p^m}\sum_{v\in\F_{p^m}^*}
    \left(\sum_{j\in\Fp}L_F(u,v;j)\zeta_p^j\right)
    \zeta_p^{-\Tr_1^m(vb_0)}\\
 &=p^{n-m}
   +\frac1{p^m}\sum_{v\in\F_{p^m}^*}
    \zeta_p^{-\Tr_1^m(vb_0)}
    \sum_{j\in\Fp}L_F(u,v;j)\zeta_p^j.
 \end{aligned}
\]
Recalling that $d_u(b_0)=\DDT_F(u,b_0)$ gives exactly~\eqref{eq:profile-ddt-inversion}.

Thus the complete set of profiles for nonzero masks determines every entry of the $u$-th DDT row.  Conversely, Proposition~\ref{prop:profile-ddt} expresses each level count $N_F(u,v;j)$ as a sum of entries of that row over the trace hyperplane $\Tr_1^m(vb)=j$; subtracting $p^{n-1}$ then determines $L_F(u,v;j)$.  Hence the two collections determine each other.
\end{proof}

Theorem~\ref{thm:profile-ddt-inversion} strengthens the hyperplane-sum interpretation: an individual profile records only the masses of a DDT row on one parallel class of trace hyperplanes, but the profiles over all nonzero masks recover the row completely.  In characteristic two, this is consistent with the Fourier relation between autocorrelation and DDT rows established in~\cite{CanteautEtAl19}.  Here the level-resolved formulation shows directly how the same information is distributed among the $p$ trace levels.  This also places the profile in the same general line of inquiry as other connectivity tables that can be controlled through DDT information~\cite{GargEtAl25BCT}, while retaining a simple linear-geometric description.

A first consequence is obtained when a derivative is balanced: uniformity on the output space forces uniformity on every nonzero trace component.

\begin{proposition}
\label{prop:balanced-derivative}
Fix $u\ne0$.  Suppose that $D_uF:\F_{p^n}\to\F_{p^m}$ is balanced, i.e., every value in $\F_{p^m}$ has exactly $p^{n-m}$ preimages.  Then, for every $v\in\F_{p^m}^*$ and every $j\in\Fp$,
 $L_F(u,v;j)=0$.
Equivalently, $\AC_F(u,tv)=0$ for all $t\in\Fp^*$, $v\in\F_{p^m}^*$.
\end{proposition}

\begin{proof}
Fix $v\in\F_{p^m}^*$.  The map
\[
 \ell_v:\F_{p^m}\longrightarrow\Fp,
 \qquad \ell_v(b)=\Tr_1^m(vb),
\]
is a nonzero $\Fp$-linear functional.  It is therefore surjective, its kernel has dimension $m-1$, and each fiber $\ell_v^{-1}(j)$ contains exactly $p^{m-1}$ elements.

By hypothesis, every $b\in\F_{p^m}$ has exactly $p^{n-m}$ preimages under $D_uF$.  Hence, for any $j\in\Fp$,
\begin{align*}
 N_F(u,v;j)
 &=\sum_{\substack{b\in\F_{p^m}\\ \Tr(vb)=j}}
   \#\{x\in\F_{p^n}:D_uF(x)=b\}\\
 &=p^{m-1}p^{n-m}=p^{n-1}.
\end{align*}
Thus $L_F(u,v;j)=0$ for every $j$.  To obtain the autocorrelation statement, fix $t\in\Fp^*$.  Formula~\eqref{eq:ac-from-profile} now reads
\[
 \AC_F(u,tv)
 =\sum_{j\in\Fp}L_F(u,v;j)\zeta_p^{tj}
 =\sum_{j\in\Fp}0\cdot\zeta_p^{tj}
 =0.
\]
Since $t$ was arbitrary, this holds for every $t\in\Fp^*$, which proves the equivalent formulation.
\end{proof}

The DDT relation becomes substantially stronger after taking second moments.  For a fixed derivative direction $u$, write
\begin{equation*}
 E_F(u)=\sum_{b\in\F_{p^m}}\DDT_F(u,b)^2.
\end{equation*}
This is the collision energy of the $u$-th DDT row.

\begin{theorem}
\label{thm:profile-ddt-energy}
For every $u\in\F_{p^n}$,
\begin{align}
 \sum_{v\in\F_{p^m}^*}\sum_{j\in\Fp}
 |L_F(u,v;j)|^2
 &=\frac{p-1}{p}
 \left(p^m E_F(u)-p^{2n}\right) \label{eq:profile-ddt-energy}\\
 &=(p-1)p^{m-1}
 \sum_{b\in\F_{p^m}}
 \left(\DDT_F(u,b)-p^{n-m}\right)^2. \label{eq:profile-ddt-distance}
\end{align}
More generally, if
\begin{equation}
\label{eq:row-diff-spectrum}
 \nu_k(u)=\#\{b\in\F_{p^m}:\DDT_F(u,b)=k\},
\end{equation}
denotes the differential spectrum of the $u$-th DDT row, then
\begin{equation}
\label{eq:profile-row-spectrum}
 \sum_{v\in\F_{p^m}^*}\sum_{j\in\Fp}|L_F(u,v;j)|^2
 =(p-1)p^{m-1}
 \sum_{k\ge0}\nu_k(u)\left(k-p^{n-m}\right)^2.
\end{equation}
Consequently, for every nonzero $u$,
\begin{equation}
\label{eq:gamma-energy-lower}
 (\Gamma_F^{(p)})^2
 \ge
 \frac{p-1}{p^2(p^m-1)}
 \left(p^m E_F(u)-p^{2n}\right),
\end{equation}
and hence the same bound holds with the right-hand side maximized over $u\ne0$.
\end{theorem}

\begin{proof}
Fix $u\in\F_{p^n}$ and set $S(u)=\sum_{v\in\F_{p^m}^*}\sum_{j\in\Fp}|L_F(u,v;j)|^2$. For each fixed $v\ne0$, Proposition~\ref{prop:parseval-profile} gives
\[
 \sum_{j\in\Fp}|L_F(u,v;j)|^2
 =\frac1p\sum_{t\in\Fp^*}|\AC_F(u,tv)|^2.
\]
Hence
\begin{align*}
 S(u)
 &=\frac1p\sum_{v\in\F_{p^m}^*}
   \sum_{t\in\Fp^*}|\AC_F(u,tv)|^2.
\end{align*}
For a fixed $w\in\F_{p^m}^*$ and each $t\in\Fp^*$ there is a unique
$v=t^{-1}w\in\F_{p^m}^*$ satisfying $w=tv$.  Thus every nonzero $w$ occurs exactly $p-1$ times in the double sum, and
\begin{equation}
\label{eq:energy-orbit-count}
 S(u)=\frac{p-1}{p}
 \sum_{w\in\F_{p^m}^*}|\AC_F(u,w)|^2.
\end{equation}

We now evaluate the full autocorrelation energy.  Since complex conjugation sends $\zeta_p^a$ to $\zeta_p^{-a}$,
\begin{align*}
 |\AC_F(u,w)|^2
 &=\left(\sum_{x\in\F_{p^n}}
   \zeta_p^{\Tr_1^m(wD_uF(x))}\right)
   \left(\sum_{y\in\F_{p^n}}
   \zeta_p^{-\Tr_1^m(wD_uF(y))}\right)\\
 &=\sum_{x,y\in\F_{p^n}}
   \zeta_p^{\Tr_1^m(w(D_uF(x)-D_uF(y)))}.
\end{align*}
Summing over $w\in\F_{p^m}$ and interchanging the finite sums gives
\begin{align*}
 \sum_{w\in\F_{p^m}}|\AC_F(u,w)|^2
 &=\sum_{x,y\in\F_{p^n}}
   \sum_{w\in\F_{p^m}}
   \zeta_p^{\Tr_1^m(w(D_uF(x)-D_uF(y)))}.
\end{align*}
By nondegeneracy of the trace pairing and additive-character orthogonality, the inner sum is $p^m$ when
$D_uF(x)=D_uF(y)$ and is $0$ otherwise.  Therefore
\begin{equation}
\label{eq:ac-energy-collisions}
 \sum_{w\in\F_{p^m}}|\AC_F(u,w)|^2
 =p^m\#\{(x,y)\in\F_{p^n}^2:D_uF(x)=D_uF(y)\}.
\end{equation}
For each $b\in\F_{p^m}$, there are exactly $\DDT_F(u,b)$ choices of $x$ and independently the same number of choices of $y$ with derivative value $b$.  Consequently,
\[
 \#\{(x,y):D_uF(x)=D_uF(y)\}
 =\sum_{b\in\F_{p^m}}\DDT_F(u,b)^2
 =E_F(u).
\]
The identity~\eqref{eq:ac-energy-collisions} thus becomes $\sum_{w\in\F_{p^m}}|\AC_F(u,w)|^2=p^mE_F(u)$. Since $\AC_F(u,0)=p^n$, the $w=0$ term has squared magnitude $p^{2n}$, and hence $\sum_{w\in\F_{p^m}^*}|\AC_F(u,w)|^2 =p^mE_F(u)-p^{2n}$. Combining this identity with~\eqref{eq:energy-orbit-count} gives
\[
 \begin{aligned}
 S(u)
 &=\frac{p-1}{p}
   \sum_{w\in\F_{p^m}^*}|\AC_F(u,w)|^2\\
 &=\frac{p-1}{p}\bigl(p^mE_F(u)-p^{2n}\bigr),
 \end{aligned}
\]
which is~\eqref{eq:profile-ddt-energy}.

We next rewrite the result in centered form.  Put $d_b=\DDT_F(u,b)$.  Because the DDT row counts all $p^n$ inputs, $\sum_{b\in\F_{p^m}}d_b=p^n$. Expanding the square gives
\begin{align*}
 \sum_b(d_b-p^{n-m})^2
 &=\sum_b d_b^2
   -2p^{n-m}\sum_b d_b
   +\sum_b p^{2n-2m}\\
 &=E_F(u)-2p^{2n-m}+p^m p^{2n-2m}\\
 &=E_F(u)-p^{2n-m}.
\end{align*}
Multiplying by $(p-1)p^{m-1}$ yields
\[
 (p-1)p^{m-1}\sum_b(d_b-p^{n-m})^2
 =\frac{p-1}{p}\left(p^mE_F(u)-p^{2n}\right),
\]
which proves~\eqref{eq:profile-ddt-distance}.

For~\eqref{eq:profile-row-spectrum}, partition the output values according to the value of the DDT entry.  By definition of $\nu_k(u)$,
\[
 \sum_{b\in\F_{p^m}}(d_b-p^{n-m})^2
 =\sum_{k\ge0}\nu_k(u)(k-p^{n-m})^2.
\]
Inserting this partition into~\eqref{eq:profile-ddt-distance} gives
\[
 \begin{aligned}
 S(u)
 &=(p-1)p^{m-1}
   \sum_{b\in\F_{p^m}}(d_b-p^{n-m})^2\\
 &=(p-1)p^{m-1}
   \sum_{k\ge0}\nu_k(u)(k-p^{n-m})^2,
 \end{aligned}
\]
which is~\eqref{eq:profile-row-spectrum}.

Finally, for fixed nonzero $u$, the left-hand side of~\eqref{eq:profile-ddt-energy} contains exactly
$p(p^m-1)$ terms: there are $p^m-1$ nonzero masks $v$ and $p$ trace levels $j$.  Each term satisfies
$|L_F(u,v;j)|^2\le(\Gamma_F^{(p)})^2$.  Therefore $S(u)\le p(p^m-1)(\Gamma_F^{(p)})^2$. On the other hand,~\eqref{eq:profile-ddt-energy} gives $S(u)=\frac{p-1}{p}\bigl(p^mE_F(u)-p^{2n}\bigr)$. Together with the preceding upper bound for $S(u)$, this yields
\[
 \frac{p-1}{p}\bigl(p^mE_F(u)-p^{2n}\bigr)
 \le p(p^m-1)(\Gamma_F^{(p)})^2.
\]
Dividing by $p(p^m-1)$ and rearranging gives
\[
 (\Gamma_F^{(p)})^2
 \ge
 \frac{p-1}{p^2(p^m-1)}
 \bigl(p^mE_F(u)-p^{2n}\bigr),
\]
which is~\eqref{eq:gamma-energy-lower}.  Since this holds for each nonzero $u$, maximizing the right-hand side over $u\ne0$ gives the final assertion.
\end{proof}

The identity gives the converse to Proposition~\ref{prop:balanced-derivative} and shows that zero profile is exactly the balanced-derivative case, not merely a consequence of it.

\begin{corollary}
\label{cor:balanced-iff-zero-profile}
Fix $u\in\F_{p^n}$.  The following are equivalent:
\begin{enumerate}[label=\textup{(\roman*)}]
\item $D_uF$ is balanced;
\item $L_F(u,v;j)=0$ for every $v\in\F_{p^m}^*$ and every $j\in\Fp$;
\item $\AC_F(u,w)=0$ for every $w\in\F_{p^m}^*$.
\end{enumerate}
\end{corollary}

\begin{proof}
Assume first that $D_uF$ is balanced.  Proposition~\ref{prop:balanced-derivative} then gives
$L_F(u,v;j)=0$ for every nonzero $v \in \F_{p^m}$ and every $j \in \F_p$, proving (i)$\Rightarrow$(ii).

Assume (ii).  For any $w\in\F_{p^m}^*$, take $v=w$ and $t=1$ in~\eqref{eq:ac-from-profile}.  Then $\AC_F(u,w)=0$, so (ii)$\Rightarrow$(iii).

Finally, we assume (iii).  The inversion formula~\eqref{eq:ddt-fourier-inversion}, proved in Theorem~\ref{thm:profile-ddt-inversion}, gives for every $b\in\F_{p^m}$
\[
 \DDT_F(u,b)
 =\frac1{p^m}\sum_{w\in\F_{p^m}}
 \AC_F(u,w)\zeta_p^{-\Tr_1^m(wb)}.
\]
All nonzero terms vanish by (iii).  The remaining term is
\[
 \DDT_F(u,b)=\frac{\AC_F(u,0)}{p^m}=\frac{p^n}{p^m}=p^{n-m},
\]
independently of $b$.  Thus every value of $\F_{p^m}$ has exactly $p^{n-m}$ preimages under $D_uF$, so $D_uF$ is balanced.  This proves (iii)$\Rightarrow$(i).
\end{proof}

\begin{corollary}
\label{cor:planar-characterization}
Let $p$ be odd and $F:\F_{p^n}\to\F_{p^n}$.  Then $F$ is planar if and only if
\[
 \mathcal L_F(u,v)=(0,\ldots,0)
 \qquad\text{for every }u,v\in\F_{p^n}^*.
\]
\end{corollary}

\begin{proof}
Suppose first that $F$ is planar.  By definition, for every $u\in\F_{p^n}^*$ the derivative
$D_uF:x\mapsto F(x+u)-F(x)$ is a permutation of $\F_{p^n}$.  Hence every output value has exactly one preimage, so $D_uF$ is balanced.  Proposition~\ref{prop:balanced-derivative} then gives $L_F(u,v;j)=0 \qquad(v\in\F_{p^n}^*,\ j\in\Fp)$, and therefore $\mathcal L_F(u,v)=(0,\ldots,0)$ for every nonzero $u,v$.

Conversely, assume that all these nontrivial profiles vanish.  Fix $u\ne0$.  Corollary~\ref{cor:balanced-iff-zero-profile} shows that $D_uF$ is balanced.  Here both domain and codomain have cardinality $p^n$, so the common fiber size of a balanced map is
$p^{n-n}=1$.  Thus $D_uF$ is bijective.  Since this holds for every $u\ne0$, $F$ is planar.
\end{proof}

We next consider what the preceding identities imply, and do not imply, for functions of low differential uniformity.

The DDT relation also clarifies what can, and cannot, be inferred from differential uniformity alone.  If $F$ is $\delta$-differentially uniform, then every summand in~\eqref{eq:profile-ddt} is at most $\delta$, and each affine trace hyperplane contains $p^{m-1}$ elements.  Hence
\begin{equation*}
\label{eq:rough-du-bound}
 0\le N_F(u,v;j)\le \delta p^{m-1}
\end{equation*}
for $u\ne0$ and $v\ne0$.  Consequently,
\begin{equation}
\label{eq:rough-profile-bound}
 -p^{n-1}\le L_F(u,v;j)\le \delta p^{m-1}-p^{n-1}.
\end{equation}
These bounds are deliberately crude: differential uniformity controls only the largest individual DDT entry, whereas the differential-linear profile depends on the whole derivative distribution.  Theorem~\ref{thm:profile-ddt-energy} gives a precise answer at the second-moment level.  The total profile energy is exactly determined by the distance of the full DDT row from the balanced row, or equivalently by the row differential spectrum~\eqref{eq:row-diff-spectrum}.  Thus replacing $\delta_F$ by the full differential spectrum does not merely sharpen~\eqref{eq:rough-profile-bound}; it determines the total squared profile values exactly.  What is not determined by the row spectrum alone is how this energy is distributed among the individual trace hyperplanes and hence the precise value of $\Gamma_F^{(p)}$.

For APN square maps in odd characteristic the row-spectrum identity has a particularly simple form.

\begin{corollary}
\label{cor:apn-profile-energy}
Let $p$ be odd and let $F:\F_{p^n}\to\F_{p^n}$ be APN.  For $u\ne0$, set $r_u=\#\{b\in\F_{p^n}:\DDT_F(u,b)=2\}$. Then the $u$-th DDT row contains exactly $r_u$ zero entries, exactly $r_u$ entries equal to $2$, and all remaining entries are equal to $1$.  Moreover,
\begin{equation}
\label{eq:apn-profile-energy}
 \sum_{v\in\F_{p^n}^*}\sum_{j\in\Fp}|L_F(u,v;j)|^2
 =2(p-1)p^{n-1}r_u.
\end{equation}
Consequently,
\begin{equation}
\label{eq:apn-gamma-lower}
 (\Gamma_F^{(p)})^2
 \ge
 \frac{2(p-1)p^{n-2}r_u}{p^n-1}.
\end{equation}
In particular, the total profile energy in direction $u$ vanishes if and only if $r_u=0$, equivalently if and only if $D_uF$ is a permutation.
\end{corollary}

\begin{proof}
Fix $u\ne0$.  Because $F$ is APN, its differential uniformity is $2$, so every entry in the $u$-th DDT row belongs to $\{0,1,2\}$.  Let
\[
 a_u=\#\{b:\DDT_F(u,b)=0\},\qquad
 b_u=\#\{b:\DDT_F(u,b)=1\},
\]
and retain the notation $r_u=\#\{b:\DDT_F(u,b)=2\}$. There are $p^n$ output values, hence
\begin{equation}
\label{eq:apn-count-1}
 a_u+b_u+r_u=p^n.
\end{equation}
On the other hand, the entries in any DDT row count all $p^n$ inputs, so
\begin{equation}
\label{eq:apn-count-2}
 0\cdot a_u+1\cdot b_u+2\cdot r_u=p^n.
\end{equation}
Subtracting~\eqref{eq:apn-count-2} from~\eqref{eq:apn-count-1} gives
$a_u=r_u$.  Thus the row contains equally many zeros and twos.

Because $m=n$, the balanced row has constant value $p^{n-m}=1$.  Applying~\eqref{eq:profile-ddt-distance} gives
\begin{align*}
 \sum_{v\in\F_{p^n}^*}\sum_{j\in\Fp}|L_F(u,v;j)|^2
 &=(p-1)p^{n-1}
   \sum_{b\in\F_{p^n}}(\DDT_F(u,b)-1)^2.
\end{align*}
For a DDT entry equal to $0$ or $2$, the squared deviation from $1$ is $1$; for an entry equal to $1$, it is $0$.  Since there are $a_u+r_u=2r_u$ entries of the first two types, $\sum_{b\in\F_{p^n}}(\DDT_F(u,b)-1)^2=2r_u$. Therefore
\[
 \begin{aligned}
 \sum_{v\in\F_{p^n}^*}\sum_{j\in\Fp}|L_F(u,v;j)|^2
 &=(p-1)p^{n-1}
   \sum_{b\in\F_{p^n}}(\DDT_F(u,b)-1)^2\\
 &=(p-1)p^{n-1}(2r_u)
 =2(p-1)p^{n-1}r_u,
 \end{aligned}
\]
which is~\eqref{eq:apn-profile-energy}.

There are $p(p^n-1)$ quantities $|L_F(u,v;j)|^2$ in this sum, each bounded above by $(\Gamma_F^{(p)})^2$.  Hence $2(p-1)p^{n-1}r_u \le p(p^n-1)(\Gamma_F^{(p)})^2$, and division by $p(p^n-1)$ yields~\eqref{eq:apn-gamma-lower}.

Finally, the profile energy vanishes if and only if $r_u=0$.  In that case $a_u=0$ as well, so every DDT entry in the row is $1$.  This means precisely that every output value has one preimage under $D_uF$, i.e., $D_uF$ is a permutation.  Conversely, if $D_uF$ is a permutation, then every $b\in\F_{p^n}$ has exactly one preimage.  Thus $\DDT_F(u,b)=1$ for all $b$, so there are no entries equal to $2$, i.e., $r_u=0$.  Formula~\eqref{eq:apn-profile-energy} then gives that total profile energy vanishes.
\end{proof}

This specialization makes the distinction between differential uniformity and differential-linear behavior explicit.  Within the APN class, $\delta_F$ is fixed, while the profile energy of a derivative direction varies with the number of double values in that derivative.  Determining the individual profile entries requires the finer placement of those multiplicities relative to the trace hyperplanes.  This is analogous to the phenomenon observed for second-order zero differential and boomerang-type criteria, where low differential uniformity does not by itself force optimal behavior for the refined statistic~\cite{GargEtAl25SOZD,GargEtAl25JAA}.  In particular, the odd-characteristic results in~\cite{GargEtAl25JAA} show that even within low-differential-uniformity classes the second-order zero differential behavior can vary substantially.

The preceding results reduce the determination of individual profiles to a finer question than differential uniformity alone: one must determine how the DDT multiplicities are distributed among the affine trace hyperplanes.  We now illustrate this point with two odd-characteristic power functions whose derivatives can be analyzed directly.  For the monomial $x^{p^k+1}$ the complete profile can be determined from a linearized derivative.  For the inverse monomial, the profile admits an exact reduction to classical Kloosterman sums.


\begin{theorem}
\label{thm:gold-profile}
Let $p$ be odd, let $k\ge1$, let $F(x)=x^{p^k+1}$ on $\F_{p^n}$, and let $e=\gcd(k,n), r=\frac ne$. For $u\in\F_{p^n}^*$ define the $\Fp$-linear map $\ell_u(x)=ux^{p^k}+u^{p^k}x$ and
\[
 \bigl(\operatorname{Im}\ell_u\bigr)^\perp
 =
 \left\{z\in\F_{p^n}:
 \Tr_1^n(zw)=0\text{ for every }w\in\operatorname{Im}\ell_u
 \right\}.
\]
Then, for every $(u,v)\in(\F_{p^n}^*)^2$, the following hold.
\begin{enumerate}[label=\textup{(\roman*)}]
\item If $r$ is odd, then $\mathcal L_F(u,v)=(0,\ldots,0)$. \item If $r$ is even, then, for every $j\in\Fp$,
\begin{equation}
\label{eq:gold-profile}
 L_F(u,v;j)=
 \begin{cases}
 (p-1)p^{n-1},
 &\text{if }v\in(\operatorname{Im}\ell_u)^\perp
   \text{ and }j=\Tr_1^n(vu^{p^k+1}),\\
 -p^{n-1},
 &\text{if }v\in(\operatorname{Im}\ell_u)^\perp
   \text{ and }j\ne\Tr_1^n(vu^{p^k+1}),\\
 0,
 &\text{if }v\notin(\operatorname{Im}\ell_u)^\perp.
 \end{cases}
\end{equation}
Moreover, in the even case, among the $(p^n-1)^2$ pairs
$(u,v)\in(\F_{p^n}^*)^2$, the profile is the zero vector for $(p^n-1)(p^n-p^e)$ pairs, and for each of the remaining $(p^n-1)(p^e-1)$ pairs it is a permutation of
\[
 \bigl((p-1)p^{n-1},
       \underbrace{-p^{n-1},\ldots,-p^{n-1}}_{p-1}\bigr).
\]
\end{enumerate}
\end{theorem}

\begin{proof}
For $u\ne0$,
\begin{align*}
 D_uF(x)
 &=(x+u)^{p^k+1}-x^{p^k+1}
 =ux^{p^k}+u^{p^k}x+u^{p^k+1}
 =\ell_u(x)+u^{p^k+1}.
\end{align*}
Thus every nonempty fiber of $D_uF$ is a coset of $\ker\ell_u$.

We first determine the kernel.  Writing $x=uy$, the equation
$\ell_u(x)=0$ becomes $y^{p^k}+y=0$. Besides $y=0$, a solution exists precisely when $y^{p^k-1}=-1$. Since $\F_{p^n}^*$ is cyclic and $\gcd(p^k-1,p^n-1)=p^e-1$, the equation $z^{p^k-1}=-1$ is solvable in $\F_{p^n}^*$ if and only if $(-1)^{(p^n-1)/(p^e-1)}=1$. Now $\frac{p^n-1}{p^e-1} =1+p^e+\cdots+p^{e(r-1)}$. Because $p$ is odd, this integer has the same parity as $r$.  Hence, if
$r$ is odd, $\ker\ell_u=\{0\}$; if $r$ is even, the equation has
$p^e-1$ nonzero solutions, so $|\ker\ell_u|=p^e$. 
Suppose first that $r$ is odd.  Then $\ell_u$ is a permutation of
$\F_{p^n}$, and therefore so is $D_uF$.  Thus every DDT entry in the
$u$-row equals $1$.  By Proposition~\ref{prop:profile-ddt}, for every
$v\ne0$ and $j\in\Fp$, $N_F(u,v;j) =\#\{b\in\F_{p^n}:\Tr_1^n(vb)=j\} =p^{n-1}$, and consequently $\mathcal L_F(u,v)=(0,\ldots,0)$.

Suppose now that $r$ is even.  The image of $\ell_u$ has cardinality
$p^{n-e}$, and
\begin{equation*}
 \DDT_F(u,b)=
 \begin{cases}
 p^e,&b\in u^{p^k+1}+\operatorname{Im}\ell_u,\\
 0,&b\notin u^{p^k+1}+\operatorname{Im}\ell_u.
 \end{cases}
\end{equation*}
Using Proposition~\ref{prop:profile-ddt} and writing
$b=u^{p^k+1}+w$, we obtain
\begin{align}
 N_F(u,v;j)
 &=p^e
 \#\left\{w\in\operatorname{Im}\ell_u:
 \Tr_1^n(vw)
 =j-\Tr_1^n(vu^{p^k+1})\right\}.
 \label{eq:gold-level-count}
\end{align}

If $v\in(\operatorname{Im}\ell_u)^\perp$, the trace on the right-hand
side of~\eqref{eq:gold-level-count} is identically zero.  Hence
$N_F(u,v;j)=p^n$ when
$j=\Tr_1^n(vu^{p^k+1})$, and it is $0$ for every other $j$.  Centering
by $p^{n-1}$ gives the first two cases of~\eqref{eq:gold-profile}.

If $v\notin(\operatorname{Im}\ell_u)^\perp$, then $w\longmapsto\Tr_1^n(vw)$ is a nonzero $\Fp$-linear map from $\operatorname{Im}\ell_u$ to
$\Fp$.  It is therefore surjective, and every fiber has
$p^{n-e-1}$ elements.  Equation~\eqref{eq:gold-level-count} then gives
$N_F(u,v;j)=p^{n-1}$ for every $j$, hence
$L_F(u,v;j)=0$.

Finally, the trace pairing $(z,w)\longmapsto\Tr_1^n(zw)$ is nondegenerate on $\F_{p^n}$, so
\[
 \dim_{\Fp}(\operatorname{Im}\ell_u)^\perp
 =n-\dim_{\Fp}(\operatorname{Im}\ell_u)=e.
\]
Thus, for each fixed $u\ne0$, exactly $p^e-1$ nonzero masks lie in the
orthogonal complement and exactly $p^n-p^e$ nonzero masks lie outside
it.  Multiplying these two numbers by the $p^n-1$ choices of $u$
gives the asserted multiplicities.
\end{proof}

The preceding theorem gives a complete numerical profile spectrum for
this Gold-type family.  The inverse map behaves differently: its DDT
row is still explicit, but the trace-hyperplane sums naturally produce Kloosterman sums, as expected.

\begin{theorem}\label{thm:inverse-profile}
Let $p$ be odd, put $q=p^n$, and let $F(x)=x^{q-2}$ on $\F_q$, with $F(0)=0$.  Let $\chi$ be the quadratic character of
$\F_q$, extended by $\chi(0)=0$, and let $\Psi(z)=\zeta_p^{\Tr_1^n(z)}$ be the canonical additive character.  For $z\in\F_q^*$, we define the Kloosterman sum
\begin{equation*}
 \Kl(z)=\sum_{s\in\F_q^*}\Psi\left(s+\frac zs\right).
\end{equation*}
Then
\begin{equation}
\label{eq:inverse-ddt-row}
 \DDT_F(1,c)=
 \begin{cases}
 0,&c=0,\\
 3+\chi(-3),&c=1,\\
 1+\chi(c^2-4c),&c\notin\{0,1\},
 \end{cases}
\end{equation}
and  
\begin{align}
 L_F(u,v;j)
 &=-\1_{\{j=0\}}
   +2\cdot \1_{\{j=\Tr_1^n(\lambda)\}}
   -\frac1p 
   +\frac1p\sum_{\alpha\in\Fp^*}
      \zeta_p^{-\alpha j}
      \Psi(2\alpha\lambda)
      \Kl(\alpha^2\lambda^2),
 \label{eq:inverse-profile-kloosterman}
\end{align}
where, for $u,v\in\F_q^*$ and $j\in\Fp$, we write $\lambda=vu^{-1}$.
\end{theorem}

\begin{proof}
For $u\ne0$, the substitution $x=uy$ gives $D_uF(uy)=u^{-1}D_1F(y)$, and therefore
\begin{equation}
\label{eq:inverse-ddt-scaling}
 \DDT_F(u,b)=\DDT_F(1,bu).
\end{equation}
For $x=0$ and $x=-1$ one has $D_1F(x)=1$.  If
$x\notin\{0,-1\}$, then $D_1F(x)=\frac1{x+1}-\frac1x =-\frac1{x(x+1)}$. Thus $D_1F(x)=c$ is equivalent to $cx^2+cx+1=0$. For $c\ne0$ this quadratic has
$1+\chi(c^2-4c)$ roots in $\F_q$.  When $c=1$, the two exceptional
inputs $0$ and $-1$ must be added, giving $3+\chi(-3)$.  When $c=0$
there is no solution.  Putting all of this together, it proves the claim~\eqref{eq:inverse-ddt-row}.

Fix now $u,v\ne0$, set $\lambda=vu^{-1}$, and use
$c=bu$ in Proposition~\ref{prop:profile-ddt}.  Combining Equations~\eqref{eq:inverse-ddt-scaling} and \eqref{eq:inverse-ddt-row} we infer
\begin{align*}
 N_F(u,v;j)
 &=
 \sum_{\substack{c\in\F_q^*\\
                  \Tr_1^n(\lambda c)=j}}
 \bigl(1+\chi(c^2-4c)\bigr)
 +2\cdot \1_{\{j=\Tr_1^n(\lambda)\}}.
 \label{eq:inverse-level-count}
\end{align*}
Because $\lambda\ne0$, the map
$c\mapsto\Tr_1^n(\lambda c)$ is a nonzero $\Fp$-linear functional on
$\F_q$.  Hence
\[
 \#\{c\in\F_q^*:\Tr_1^n(\lambda c)=j\}
 =p^{n-1}-\1_{\{j=0\}}.
\]
Consequently,
\begin{equation}
\label{eq:inverse-profile-mixed}
 L_F(u,v;j)
 =-\1_{\{j=0\}}
  +2\cdot \1_{\{j=\Tr_1^n(\lambda)\}}
  +S_j(\lambda),
\end{equation}
where
\[
 S_j(\lambda)
 =
 \sum_{\substack{c\in\F_q\\
                  \Tr_1^n(\lambda c)=j}}
 \chi(c^2-4c).
\]
(The term $c=0$ may be included because $\chi(0)=0$.)

By additive-character orthogonality,
\begin{align}
 S_j(\lambda)
 &=
 \frac1p\sum_{\alpha\in\Fp}
 \zeta_p^{-\alpha j}
 \sum_{c\in\F_q}
 \chi(c^2-4c)\Psi(\alpha\lambda c).
 \label{eq:inverse-orthogonality}
\end{align}
It remains to evaluate the inner mixed sum.  The affine conic $y^2=c^2-4c$ is parametrized bijectively by $t\in\F_q^*$ through
\begin{equation*}
 c=2+t+t^{-1},
 \qquad
 y=t-t^{-1}.
\end{equation*}
Indeed, $(c-2+y)(c-2-y)=4$, and every point of the conic has the unique parameter
$t=(c-2+y)/2\ne0$.  It follows in particular that the conic has
$q-1$ affine points.  Since the number of $y$ above a fixed $c$ is
$1+\chi(c^2-4c)$, we obtain $\sum_{c\in\F_q}\chi(c^2-4c)=(q-1)-q=-1$. This evaluates the term $\alpha=0$ in Equation~\eqref{eq:inverse-orthogonality}.

For $a\in\F_q^*$, the same point count, now weighted by $\Psi(ac)$,
gives
\begin{align*}
 \sum_{c\in\F_q}\chi(c^2-4c)\Psi(ac)
 &=
 \sum_{\substack{(c,y)\in\F_q^2\\y^2=c^2-4c}}
 \Psi(ac)
 -\sum_{c\in\F_q}\Psi(ac)\\
 &=
 \sum_{t\in\F_q^*}
 \Psi\bigl(a(2+t+t^{-1})\bigr)\\
 &=
 \Psi(2a)
 \sum_{t\in\F_q^*}\Psi(at+a/t).
\end{align*}
The second term in the first line vanishes because $a\ne0$.  With
$s=at$ in the last sum,
\[
 \sum_{t\in\F_q^*}\Psi(at+a/t)
 =
 \sum_{s\in\F_q^*}\Psi\left(s+\frac{a^2}{s}\right)
 =\Kl(a^2).
\]
Therefore
\begin{equation*}
 \sum_{c\in\F_q}\chi(c^2-4c)\Psi(ac)
 =\Psi(2a)\Kl(a^2)
 \qquad(a\ne0).
\end{equation*}
Substituting $a=\alpha\lambda$ into
\eqref{eq:inverse-orthogonality} yields
\[
 S_j(\lambda)
 =
 -\frac1p
 +\frac1p\sum_{\alpha\in\Fp^*}
 \zeta_p^{-\alpha j}
 \Psi(2\alpha\lambda)
 \Kl(\alpha^2\lambda^2).
\]
Combining this with~\eqref{eq:inverse-profile-mixed} proves
Equation~\eqref{eq:inverse-profile-kloosterman}.
\end{proof}

\begin{remark}\label{rem:inverse-spectrum}
Theorem~\ref{thm:inverse-profile} is an exact profile formula, however, unlike Theorem~\ref{thm:gold-profile},  it does not by itself enumerate the numerical multiset of all
profile values.  Such an enumeration amounts to controlling the
Kloosterman sums in~\eqref{eq:inverse-profile-kloosterman} as
$\lambda$ ranges over $\F_q^*$.
\end{remark}

\section{Equivalence behavior}\label{sec:equivalence}

The full profile has a particularly natural behavior under affine and extended-affine equivalence.  This is one reason to retain all trace levels for $p>2$: an added affine output term generally translates the level $j$, rather than preserving the distinguished level $j=0$.

We use the standard vector-space formulation.  Two functions $F,G:\F_{p^n}\to\F_{p^m}$ are EA-equivalent if there exist affine permutations $A$ of $\F_{p^n}$ and $B$ of $\F_{p^m}$ and an affine map $C:\F_{p^n}\to\F_{p^m}$ such that $G=B\circ F\circ A+C$; they are affine equivalent when one may take $C=0$.  Write
\[
 A(x)=A_0(x)+a,
 \qquad
 B(y)=B_0(y)+b,
 \qquad
 C(x)=C_0(x)+c,
\]
where $A_0:\F_{p^n}\to\F_{p^n}$ and $B_0:\F_{p^m}\to\F_{p^m}$ are invertible $\Fp$-linear maps and $C_0:\F_{p^n}\to\F_{p^m}$ is $\Fp$-linear.  For the nondegenerate trace pairing on $\F_{p^m}$, let $B_0^*$ denote the adjoint of $B_0$, characterized by
\begin{equation}
\label{eq:trace-adjoint}
 \Tr_1^m\bigl(vB_0(y)\bigr)
 =\Tr_1^m\bigl(B_0^*(v)y\bigr)
 \qquad(v,y\in\F_{p^m}).
\end{equation}
Since the trace pairing is nondegenerate, $B_0^*$ is again an invertible $\Fp$-linear map.

\begin{theorem}
\label{thm:ea-profile}
Let $F,G:\F_{p^n}\to\F_{p^m}$ be EA-equivalent, with
$ G=B\circ F\circ A+C$.
Then, for every $u\in\F_{p^n}$, $v\in\F_{p^m}$, and $j\in\Fp$,
\begin{equation}
\label{eq:ea-profile}
 L_G(u,v;j)
 =L_F\!\left(A_0(u),B_0^*(v);
 j-\Tr_1^m(vC_0(u))\right).
\end{equation}
Equivalently, the autocorrelations satisfy
\begin{equation}
\label{eq:ea-ac}
 \AC_G(u,v)
 =\zeta_p^{\Tr_1^m(vC_0(u))}
 \AC_F\bigl(A_0(u),B_0^*(v)\bigr).
\end{equation}
In particular, if $F$ and $G$ are affine equivalent, so that $C=0$, then
\begin{equation}
\label{eq:affine-profile}
 L_G(u,v;j)
 =L_F\bigl(A_0(u),B_0^*(v);j\bigr).
\end{equation}
\end{theorem}

\begin{proof}
Write
\[
 A(x)=A_0(x)+a,\qquad B(y)=B_0(y)+b,\qquad C(x)=C_0(x)+c.
\]
Since $A_0$ and $B_0$ are $\Fp$-linear, for every $x,u\in\F_{p^n}$ we have $A(x+u)=A_0(x+u)+a=A(x)+A_0(u)$. Also, $C(x+u)-C(x)=C_0(u)$. Using $G=B\circ F\circ A+C$, we therefore obtain
\begin{align*}
 D_uG(x)
 &=G(x+u)-G(x)\\
 &=B\bigl(F(A(x+u))\bigr)-B\bigl(F(A(x))\bigr)
   +C(x+u)-C(x)\\
 &=B_0\bigl(F(A(x+u))-F(A(x))\bigr)+C_0(u)\\
 &=B_0\bigl(F(A(x)+A_0(u))-F(A(x))\bigr)+C_0(u)\\
 &=B_0\bigl(D_{A_0(u)}F(A(x))\bigr)+C_0(u).
\end{align*}
The constants $b$ and $c$ disappear in the second line because they cancel under subtraction.

Fix $v\in\F_{p^m}$.  By the defining property~\eqref{eq:trace-adjoint} of the trace adjoint,
\begin{align*}
 \Tr_1^m\bigl(vD_uG(x)\bigr)
 &=\Tr_1^m\bigl(vB_0(D_{A_0(u)}F(A(x)))\bigr)
   +\Tr_1^m(vC_0(u))\\
 &=\Tr_1^m\bigl(B_0^*(v)D_{A_0(u)}F(A(x))\bigr)
   +\Tr_1^m(vC_0(u)).
\end{align*}
Set $c_{u,v}=\Tr_1^m(vC_0(u))\in\Fp$. Then, for any $j\in\Fp$,
\begin{align*}
 \Tr_1^m(vD_uG(x))=j
 &\Longleftrightarrow
 \Tr_1^m\bigl(B_0^*(v)D_{A_0(u)}F(A(x))\bigr)=j-c_{u,v}.
\end{align*}
Because $A$ is an affine permutation, the change of variable $x'=A(x)$ is a bijection of $\F_{p^n}$.  Counting solutions on both sides gives $N_G(u,v;j) =N_F\bigl(A_0(u),B_0^*(v);j-c_{u,v}\bigr)$. Using the definition of the centered profile on both sides of the preceding counting identity,
\[
 \begin{aligned}
 L_G(u,v;j)
 &=N_G(u,v;j)-p^{n-1}\\
 &=N_F\bigl(A_0(u),B_0^*(v);j-c_{u,v}\bigr)-p^{n-1}\\
 &=L_F\bigl(A_0(u),B_0^*(v);j-c_{u,v}\bigr).
 \end{aligned}
\]
This is Equation~\eqref{eq:ea-profile}.

For the autocorrelation identity, substitute the derivative formula directly into the definition:
\begin{align*}
 \AC_G(u,v)
 &=\sum_{x\in\F_{p^n}}
   \zeta_p^{\Tr_1^m(vD_uG(x))}\\
 &=\sum_x
   \zeta_p^{\Tr_1^m(B_0^*(v)D_{A_0(u)}F(A(x)))+c_{u,v}}\\
 &=\zeta_p^{c_{u,v}}
   \sum_x
   \zeta_p^{\Tr_1^m(B_0^*(v)D_{A_0(u)}F(A(x)))}.
\end{align*}
Again put $x'=A(x)$.  Since $A$ is bijective,
\[
 \sum_x
 \zeta_p^{\Tr_1^m(B_0^*(v)D_{A_0(u)}F(A(x)))}
 =\sum_{x'\in\F_{p^n}}
 \zeta_p^{\Tr_1^m(B_0^*(v)D_{A_0(u)}F(x'))},
\]
which is $\AC_F(A_0(u),B_0^*(v))$.  Hence
\[
 \AC_G(u,v)
 =\zeta_p^{\Tr_1^m(vC_0(u))}
 \AC_F(A_0(u),B_0^*(v)),
\]
which is~\eqref{eq:ea-ac}.  If $C=0$, then $C_0=0$ and hence $c_{u,v}=0$ for every $(u,v)$.  The profile identity then reduces to $L_G(u,v;j) =L_F\bigl(A_0(u),B_0^*(v);j\bigr)$, which is~\eqref{eq:affine-profile}.
\end{proof}

For scalar affine maps $A(x)=\alpha x+\alpha'$ and $B(y)=\beta y+\beta'$, the adjoint of multiplication by $\beta$ is again multiplication by $\beta$, because $\Tr_1^m(v\beta y)=\Tr_1^m((\beta v)y)$. In this scalar case, the general affine transformation law specializes to $L_G(u,v;j)=L_F(\alpha u,\beta v;j)$. For $H=F+C$ with $C(x)=\gamma x+\gamma'$, Theorem~\ref{thm:ea-profile} specializes to
\begin{equation*}
 L_H(u,v;j)
 =L_F\!\left(u,v;
 j-\Tr_1^m(v\gamma u)\right),
\end{equation*}
while $\AC_H(u,v) =\zeta_p^{\Tr_1^m(v\gamma u)}\AC_F(u,v)$. In particular, the profile is translated among its $p$ levels; it is not multiplied by a root of unity.  The latter could not hold in general because the profile entries are integers.

The transformation law yields the invariant that is most natural for the level-resolved theory.

\begin{corollary}
\label{cor:ea-invariance}
If $F$ and $G$ are EA-equivalent, then the multiset
$
\left\{
L_F(u,v;j):
u\in\F_{p^n}^*,\ v\in\F_{p^m}^*,\ j\in\Fp
\right\}
$
is the same for $F$ and $G$.  In particular,
$ \Gamma_G^{(p)}=\Gamma_F^{(p)}$.
For affine equivalence, each profile is preserved after reindexing the input and output masks, without a level translation.  In particular, affine-equivalent functions also satisfy
 $\gamma_G^{(p,0)}=\gamma_F^{(p,0)}.$
\end{corollary}

\begin{proof}
Let
\[
 \Phi(u,v,j)=\bigl(A_0(u),B_0^*(v),j-\Tr_1^m(vC_0(u))\bigr).
\]
Because $A_0$ and $B_0^*$ are invertible $\Fp$-linear maps, they induce bijections of
$\F_{p^n}^*$ and $\F_{p^m}^*$, respectively.  For each fixed pair $(u,v)$, translation by the element
$-\Tr_1^m(vC_0(u))$ is a bijection of $\Fp$.  It follows that $\Phi$ is a bijection of $\F_{p^n}^*\times\F_{p^m}^*\times\Fp$. Indeed, given $(u',v',j')$, the first two coordinates uniquely determine
$u=A_0^{-1}(u')$ and $v=(B_0^*)^{-1}(v')$, and then the third coordinate uniquely determines
$j=j'+\Tr_1^m(vC_0(u))$.

By~\eqref{eq:ea-profile}, $L_G(u,v;j)=L_F(\Phi(u,v,j))$. Since $\Phi$ is bijective, the collection of all profile entries of $G$ indexed by nonzero input and output masks and by all $j\in\Fp$ is exactly the same multiset as the corresponding collection for $F$.  Taking the largest absolute value on both sides gives
$\Gamma_G^{(p)}=\Gamma_F^{(p)}$.

If the equivalence is affine, then $C=0$ and hence $C_0=0$.  Formula~\eqref{eq:affine-profile} becomes $L_G(u,v;j)=L_F(A_0(u),B_0^*(v);j)$, so the trace level itself is unchanged.  Restricting this identity to $j=0$ and using the bijectivity of $A_0$ and $B_0^*$ on the nonzero masks yields
\[
 \max_{u\ne0,\,v\ne0}|L_G(u,v;0)|
 =\max_{u'\ne0,\,v'\ne0}|L_F(u',v';0)|,
\]
which is the last claim. The proof is shown.
\end{proof}

\begin{remark}
\label{rem:zero-level-ea}
For affine equivalence, every trace level is preserved after reindexing the input and output masks, and therefore the zero-level parameter $\gamma_F^{(p,0)}$ is affine-invariant by Corollary~\ref{cor:ea-invariance}.  For general EA equivalence and $p>2$, the added affine output term may send the level $j=0$ to a nonzero level.  Thus EA-equivalence naturally preserves the full-profile parameter $\Gamma_F^{(p)}$, but not a distinguished level in isolation.  For instance, $F(x)=x^9$ and $G(x)=x^9+x$ over $\F_{5^2}$ are EA-equivalent.  A direct enumeration of the nonzero input and output masks and all five trace levels gives $\gamma_F^{(5,0)}=4$ and $\gamma_G^{(5,0)}=5$, while $\Gamma_F^{(5)}=\Gamma_G^{(5)}=5$.
The binary case is exceptional: when $p=2$, the two centered levels satisfy $L_F(u,v;1)=-L_F(u,v;0)$, so a level translation changes at most the sign and the usual absolute DLCT magnitude remains EA-invariant, in agreement with the binary theory of~\cite{CanteautEtAl19}.
\end{remark}

We shall now study the behavior under CCZ equivalence. Recall the standard graph formulation.  Let $q=p^n$ and let $F,G:\F_q\to\F_q$.  They are \emph{CCZ-equivalent} if there exists an affine permutation $\mathcal A$ of $\F_q\times\F_q$ that maps the graph of $F$ onto the graph of $G$~\cite{CCZ98}.  Write the linear part of $\mathcal A$ in block form as
\begin{equation*}
 M(x,y)=
 \bigl(M_{11}(x)+M_{12}(y),\,
       M_{21}(x)+M_{22}(y)\bigr),
\end{equation*}
where each $M_{ij}:\F_q\to\F_q$ is $\Fp$-linear.  Indeed, this decomposition is obtained by restricting $M$ to $\F_q\times\{0\}$ and $\{0\}\times\F_q$ and then projecting onto the two coordinates.  Thus $\mathcal A(x,y)=M(x,y)+(c,d)$ for some $(c,d)\in\F_q^2$.  Since $\mathcal A$ maps the graph of $F$ bijectively onto the graph of $G$, the first-coordinate map $\phi_F(x)=M_{11}(x)+M_{12}(F(x))+c$ is a permutation of $\F_q$.

As we can see in the following example, the behavior under CCZ equivalence is subtler than under EA-equivalence because a general CCZ transformation can mix an input difference with the corresponding output difference.  Nevertheless, the transport of individual DDT entries is exact, and this gives a precise formula for the induced transformation of the level counts, as shown in Theorem \ref{thm:ccz-ddt-transport}.

\begin{example}
Let $q=27$ and let $g$ be a root of the primitive polynomial
$X^{3}+2X^{2}+1\in\F_3[X]$, so that $\F_{27}=\F_3(g)$; then, $F(x)$ is a permutation of $\F_q$ with differential uniformity $3$. Taking $\mathcal A(x,y)=(y,x)$, so that
$M_{11}=M_{22}=0$ and $M_{12}=M_{21}=\mathrm{id}$, the map $\mathcal A$
carries the graph of $F$ onto the graph of $G=F^{-1}$; hence $F$ and $G$
are CCZ-equivalent, and \eqref{eq:ccz-ddt-transport} specializes to $\DDT_F(u,b)=\DDT_G(b,u)$. 

For all $u, v \in \F_q^*$ and $j \in \F_3$, $L_F(u,v;j) \equiv 0 \pmod 3$ whereas $L_{G}\bigl(1,g+g^{2}\bigr)=(-5,1,4).$ and hence the vector $(-5,1,4)$ occurs in the profile of $G$ but not in  the profile of $F$, therefore, $\{L_F(u,v) \mid u,v \in \F_q^*\} \neq \{L_G(u,v) \mid u,v \in \F_q^*\}$, and CCZ does not preserve the DLCT for general $p$ \textup{(}this was shown for $p=2$ in~\textup{\cite{LiLiLiQu19})}.
 \end{example}

In the binary setting, preservation of the differential spectrum under CCZ equivalence follows from the original graph-equivalence framework~\cite{CCZ98}.  We give the coordinate argument below over arbitrary prime characteristic, both to make the extension explicit and because the resulting transport formula is needed for the $p$-ary level counts.

\begin{theorem}\label{thm:ccz-ddt-transport}
Let $F,G:\F_q\to\F_q$ be CCZ-equivalent through $\mathcal A$ as above.  For $u,b\in\F_q$, define
\[
 \gamma(u,b)=M_{11}(u)+M_{12}(b),
 \qquad
 \delta(u,b)=M_{21}(u)+M_{22}(b).
\]
Then the DDT entry indexed by $(u,b)$ is transported according to
\begin{equation}
\label{eq:ccz-ddt-transport}
 \DDT_F(u,b)
 =
 \DDT_G\bigl(\gamma(u,b),\delta(u,b)\bigr).
\end{equation}
Consequently, if $u\in\F_q$, $v\in\F_q$, and $j\in\Fp$ are fixed, summing the transported DDT entries over the trace hyperplane
$\{c\in\F_q:\Tr_1^n(vc)=j\}$ gives
\begin{equation}
\label{eq:ccz-profile-transport}
 N_F(u,v;j)
 =
 \sum_{\substack{c\in\F_q\\ \Tr_1^n(vc)=j}}
 \DDT_G\bigl(\gamma(u,c),\delta(u,c)\bigr).
\end{equation}
After centering this level count, the corresponding profile entry is therefore
\begin{equation}
\label{eq:ccz-centered-profile-transport}
 L_F(u,v;j)
 =
 \sum_{\substack{c\in\F_q\\ \Tr_1^n(vc)=j}}
 \DDT_G\bigl(\gamma(u,c),\delta(u,c)\bigr)
 -p^{n-1}.
\end{equation}
\end{theorem}

\begin{proof}
Fix $u,b\in\F_q$ and set $S_F(u,b)=\{x\in\F_q:D_uF(x)=b\}$. For $x\in S_F(u,b)$, let $y=\phi_F(x)=M_{11}(x)+M_{12}(F(x))+c$. Because $\mathcal A$ maps the graph of $F$ onto the graph of $G$, we have  $\mathcal A(x,F(x))=(y,G(y))$.
Since $D_uF(x)=b$, the second graph point satisfies $(x+u,F(x+u))=(x,F(x))+(u,b)$. Using the affine form of $\mathcal A$ and the linearity of $M$, we obtain
\begin{align*}
 \mathcal A(x+u,F(x+u))
 =\mathcal A(x,F(x))+M(u,b)
 =(y,G(y))+
   \bigl(\gamma(u,b),\delta(u,b)\bigr).
\end{align*}
The left-hand side is another point on the graph of $G$.  Therefore
\[
 \mathcal A(x+u,F(x+u))
 =\bigl(y+\gamma(u,b),G(y+\gamma(u,b))\bigr),
\]
and comparison of the second coordinates gives
$ D_{\gamma(u,b)}G(y)=\delta(u,b)$.
Thus the permutation $\phi_F$ sends $S_F(u,b)$ injectively into $S_G\bigl(\gamma(u,b),\delta(u,b)\bigr)$.

It remains to prove surjectivity on these derivative fibers.  Let $y\in S_G\bigl(\gamma(u,b),\delta(u,b)\bigr)$.
Then both $(y,G(y))$ and 
 $\bigl(y+\gamma(u,b),G(y)+\delta(u,b)\bigr)$
lie on the graph of $G$, and their difference is
$ \bigl(\gamma(u,b),\delta(u,b)\bigr)=M(u,b)$.
Applying $\mathcal A^{-1}$ to the two graph points gives two points $(x,F(x))$ and $(x',F(x'))$ on the graph of $F$.  Because the translation part of $\mathcal A$ cancels when a difference is taken, their images satisfy
\[
 M\bigl(x'-x,F(x')-F(x)\bigr)
 =\bigl(\gamma(u,b),\delta(u,b)\bigr)
 =M(u,b).
\]
The linear map $M$ is invertible, so equality of the two images implies equality of their arguments: $x'-x=u,\qquad F(x')-F(x)=b$. Thus $x'=x+u$ and $D_uF(x)=b$, so $x\in S_F(u,b)$.  Moreover, the first coordinate of $\mathcal A(x,F(x))$ is $y$, hence $y=\phi_F(x)$.  Therefore $\phi_F$ restricts to a bijection
\[
 S_F(u,b)\longrightarrow
 S_G\bigl(\gamma(u,b),\delta(u,b)\bigr),
\]
which proves~\eqref{eq:ccz-ddt-transport}.

We now derive the profile transport formula.  Fix $u,v$, and $j$.  Proposition~\ref{prop:profile-ddt} expresses the level count as
\[
 N_F(u,v;j)
 =\sum_{\substack{c\in\F_q\\\Tr_1^n(vc)=j}}\DDT_F(u,c).
\]
For every derivative value $c$ occurring in this sum, the first part of the theorem gives $\DDT_F(u,c) =\DDT_G\bigl(\gamma(u,c),\delta(u,c)\bigr)$. Replacing each DDT entry in the hyperplane sum by its transported value therefore gives
 \begin{align*}
 N_F(u,v;j)
 =\sum_{\substack{c\in\F_q\\\Tr_1^n(vc)=j}}
   \DDT_F(u,c)
 =\sum_{\substack{c\in\F_q\\\Tr_1^n(vc)=j}}
   \DDT_G\bigl(\gamma(u,c),\delta(u,c)\bigr),
 \end{align*}
which is~\eqref{eq:ccz-profile-transport}.  By Definition~\ref{def:level-count}, the centered profile is obtained by subtracting $p^{n-1}$ from this level count.  Hence
\[
 \begin{aligned}
 L_F(u,v;j)
 =N_F(u,v;j)-p^{n-1}
 =\sum_{\substack{c\in\F_q\\\Tr_1^n(vc)=j}}
   \DDT_G\bigl(\gamma(u,c),\delta(u,c)\bigr)-p^{n-1},
 \end{aligned}
\]
which is Equation~\eqref{eq:ccz-centered-profile-transport}.
\end{proof}

Theorem~\ref{thm:ccz-ddt-transport} identifies the distinction between EA and general CCZ equivalence.  If $M_{12}=0$, then $\gamma(u,b)=M_{11}(u)$ is independent of $b$, so all terms in~\eqref{eq:ccz-profile-transport} remain in a single DDT row of $G$; this is the situation underlying the EA transformation law proved above.  For a genuine CCZ transformation with $M_{12}\ne0$, the input difference $\gamma(u,b)$ may vary with $b$.  Thus a trace-hyperplane sum in one DDT row of $F$ is generally transported to a sum involving several DDT rows of $G$.  This gives the structural obstruction to a profile-to-profile transformation law of the EA type.  In characteristic two, actual failure of CCZ invariance for the corresponding autocorrelation and DLCT quantities is known~\cite{CanteautEtAl19}.


\begin{corollary}\label{cor:ccz-global-energy}
Let $F,G:\F_q\to\F_q$ be CCZ-equivalent, where $q=p^n$.  Define the global nontrivial profile energy by
\begin{equation*}
 \mathcal E(F)=
 \sum_{u\in\F_q^*}
 \sum_{v\in\F_q^*}
 \sum_{j\in\Fp}|L_F(u,v;j)|^2.
\end{equation*}
Then  $\mathcal E(F)=\mathcal E(G)$.
\end{corollary}

\begin{proof}
Since $F$ and $G$ are square maps, Theorem~\ref{thm:profile-ddt-energy} gives, for every $u\in\F_q$,
\[
 \sum_{v\in\F_q^*}\sum_{j\in\Fp}|L_F(u,v;j)|^2
 =(p-1)p^{n-1}
 \sum_{b\in\F_q}\bigl(\DDT_F(u,b)-1\bigr)^2.
\]
Summing over all $u\in\F_q$ yields
\begin{equation}
\label{eq:all-u-energy}
\sum_{u\in\F_q}\sum_{v\in\F_q^*}\sum_{j\in\Fp}|L_F(u,v;j)|^2
 =(p-1)p^{n-1}
 \sum_{(u,b)\in\F_q^2}\bigl(\DDT_F(u,b)-1\bigr)^2.
\end{equation}
By Theorem~\ref{thm:ccz-ddt-transport}, for each $(u,b)\in\F_q^2$, $\DDT_F(u,b) =\DDT_G\bigl(\gamma(u,b),\delta(u,b)\bigr)$, and the ordered pair $(\gamma(u,b),\delta(u,b))$ is precisely $M(u,b)$.  Because $M$ is an invertible $\Fp$-linear map on $\F_q^2$, the map $(u,b)\longmapsto \bigl(\gamma(u,b),\delta(u,b)\bigr)$ is a bijection of $\F_q^2$.  Reindexing the sum by this bijection gives
\[
 \sum_{(u,b)\in\F_q^2}\bigl(\DDT_F(u,b)-1\bigr)^2
 =
 \sum_{(\gamma,\delta)\in\F_q^2}
 \bigl(\DDT_G(\gamma,\delta)-1\bigr)^2.
\]
Applying the same identity~\eqref{eq:all-u-energy} to $G$ gives
\[
 \sum_{\gamma\in\F_q}\sum_{v\in\F_q^*}\sum_{j\in\Fp}
 |L_G(\gamma,v;j)|^2
 =(p-1)p^{n-1}
 \sum_{(\gamma,\delta)\in\F_q^2}
 \bigl(\DDT_G(\gamma,\delta)-1\bigr)^2.
\]
The right-hand side equals that of~\eqref{eq:all-u-energy} for $F$ by the preceding reindexing.  Therefore
\[
 \sum_{u\in\F_q}\sum_{v\in\F_q^*}\sum_{j\in\Fp}|L_F(u,v;j)|^2
 =
 \sum_{\gamma\in\F_q}\sum_{v\in\F_q^*}\sum_{j\in\Fp}|L_G(\gamma,v;j)|^2.
\]

It remains only to remove the term $u=0$.  Since $D_0F(x)=0$ for every $x$, independently of $F$, $\DDT_F(0,0)=q, \qquad \DDT_F(0,b)=0\quad(b\ne0)$. Therefore
\[
 \sum_{b\in\F_q}\bigl(\DDT_F(0,b)-1\bigr)^2
 =(q-1)^2+(q-1)=q(q-1),
\]
and the same value holds for $G$.  Thus the $u=0$ contributions are equal and may be subtracted from the two total-energy identities.  The remaining sums are precisely $\mathcal E(F)$ and $\mathcal E(G)$, which proves the claim.
\end{proof}

\section{Conclusion}\label{S6}
We introduced a level-resolved $p$-ary differential-linear profile that retains the centered multiplicities of all trace values of a derivative component.  Its discrete Fourier transform is the family of additive autocorrelations indexed by the nonzero prime-field multiples of the output mask.  This gives the binary identity $\DLCT=\AC/2$ when $p=2$ and, for general $p$, explains why one zero-trace count does not contain the full level distribution.

The profile and the DDT are linked without loss of differential information.  For a fixed input difference, profile entries are affine trace-hyperplane sums of the corresponding DDT row, while the profiles over all nonzero masks recover that row by Fourier inversion.  The exact second-moment identity further shows that the total profile energy in a derivative direction is a constant multiple of the squared distance from the balanced DDT row.  Hence the row differential spectrum, rather than differential uniformity alone, determines the aggregate profile energy.  In particular, the profiles in a direction all vanish exactly when that derivative is balanced; for square maps in odd characteristic, vanishing in every nonzero direction is equivalent to planarity.  The two power-function examples illustrate both ends of the explicit calculation: for the Gold $x^{p^k+1}$ function, the profile spectrum is completely determined by the image of a linearized derivative, whereas for the inverse monomial the remaining trace-hyperplane sum reduces exactly to a classical Kloosterman sum.

The equivalence results show why retaining every trace level is useful.  EA-equivalence reindexes the input and output masks and translates the level, and therefore preserves the complete profile spectrum and $\Gamma_F^{(p)}$.  A general CCZ equivalence can mix several derivative directions, so the same entrywise reindexing does not hold, but the global nontrivial profile energy remains invariant for square maps.  Natural next problems are to determine closed numerical spectra for further odd-characteristic families, to identify which finer profile statistics are CCZ-invariant, and to connect the levelwise quantities more directly with nonbinary differential-linear attacks.

\end{document}